%% file: main.tex
\documentclass[conference]{IEEEtran}
\IEEEoverridecommandlockouts 

\usepackage{graphicx} 
\usepackage{pgfplots}
\usepackage{pgfplotstable}
\usepackage{blindtext}
\usepackage{amssymb}
\usepackage{amsmath}
\usepackage{amsthm}
\usepackage{bm}
\usepackage{tikz}
\usetikzlibrary{calc}
\usetikzlibrary{arrows.meta, positioning, fit, backgrounds}
\usepackage{pgfplots}
\pgfplotsset{compat=newest}
\usepackage{stfloats}  
\usepackage{algorithm}
\usepackage[noend]{algcompatible} 
\usepackage{algpseudocode}        
\usepackage{glossaries}
\usepackage{float}
\usepackage{xcolor}
\usepackage{siunitx}
\usepackage{url} 
\usepackage{cite}
\usepackage{enumitem}
\usepackage{bbm}
\usepackage{xcolor}
\usepackage{multibib}

\usepackage{comment}
\usepackage{todonotes}

\makeatletter
\newcommand\fs@spaceruled{\def\@fs@cfont{\bfseries}\let\@fs@capt\floatc@ruled
\def\@fs@pre{\vspace{.05in}\hrule height.8pt depth0pt \kern2pt}%
\def\@fs@post{\kern2pt\hrule\relax}%
\def\@fs@mid{\kern2pt\hrule\kern2pt}%
\let\@fs@iftopcapt\iftrue}
\makeatother
 
\floatstyle{spaceruled}
\restylefloat{algorithm}

\newtheorem{proposition}{Proposition}

\tikzset{
	problem/.style={
		rectangle,
		draw,
		rounded corners,
		minimum width=1.8cm,
		minimum height=0.8cm,
		align=center
	},
	arrow/.style={-{Latex}, thick},
	doublearrow/.style={<->, thick}
}

\input{acronyms}

\glsenableentrycount

\definecolor{myblue}{RGB}{0,104,180}
\definecolor{myred}{RGB}{244,161,152}
\definecolor{mydarkred}{RGB}{157,34,70}
\definecolor{mygreen}{RGB}{0,136,120}
\definecolor{mypantone}{RGB}{59,41,106}

\newcommand\copyrighttext{%
	\fontsize{6}{7}\selectfont \textcopyright This work has been submitted to the IEEE for possible publication. Copyright may be transferred without notice, after which this version may no longer be accessible.}
\newcommand\copyrightnotice{%
	\begin{tikzpicture}[remember picture,overlay]
		\node[anchor=south,yshift=20pt] at (current page.south) {\parbox{\textwidth}{\copyrighttext}};
	\end{tikzpicture}%
}

\begin{document}

\title{Dynamic Downlink-Uplink Spectrum Sharing between Terrestrial and Non-Terrestrial Networks}

\author{\IEEEauthorblockN{Sourav Mukherjee\IEEEauthorrefmark{1}, Bho Matthiesen\IEEEauthorrefmark{1}, Armin Dekorsy\IEEEauthorrefmark{1}, Petar Popovski\IEEEauthorrefmark{2}\IEEEauthorrefmark{1}}
	\IEEEauthorblockA{\IEEEauthorrefmark{1}University of Bremen, Department of Communications Engineering, Germany\\\IEEEauthorrefmark{2}Aalborg University, Department of Electronic Systems, Denmark\\ email: \{mukherjee, matthiesen, dekorsy\}@ant.uni-bremen.de, petarp@es.aau.dk}
	\thanks{
		This work is supported by the German Research Foundation (DFG) under Grant EXC 2077 (University Allowance).
	}%
}

\maketitle
\copyrightnotice

	\begin{abstract}
6G networks are expected to integrate low Earth orbit satellites to ensure global connectivity by extending coverage to underserved and remote regions.
However, the deployment of dense mega-constellations introduces severe interference among satellites operating over shared frequency bands.
This is, in part, due to the limited flexibility of conventional frequency division duplex (FDD) systems, where fixed bands for downlink (DL) and uplink (UL) transmissions are employed.
In this work, we propose dynamic re-assignment of FDD bands for improved interference management in dense deployments and evaluate the performance gain of this approach.
To this end, we formulate a joint optimization problem that incorporates dynamic band assignment, user scheduling, and power allocation in both directions.
This non-convex mixed integer problem is solved using a combination of equivalence transforms, alternating optimization, and state-of-the-art industrial-grade mixed integer solvers.
Numerical results demonstrate that the proposed approach of dynamic FDD band assignment significantly enhances system performance over conventional FDD, achieving up to 94\,\% improvement in throughput in dense deployments.
	\end{abstract}
	\begin{IEEEkeywords}
		Spectrum sharing, Dynamic downlink-uplink band, LEO, interference-mitigation.
	\end{IEEEkeywords}

\section{Introduction}
The \gls{ntn} is pivotal for realizing the long-standing vision of global coverage and bridging the digital divide~\cite{Yaacoub2020}. Among the others, \gls{leo} constellations are envisioned as a cornerstone of future technologies, thanks to their lower launching cost, lower path-loss, and delay. Further, due to ambitious projects by Iridium, OneWeb, SpaceX, Amazon, and Europe's IRIS$^2$, the \cgls{leo} will likely become crowded. Forcing multiple satellites in different orbits to share the frequency bands. 

To enable spectrum sharing, the early works focus in between the \gls{geo} and \cgls{ngeo} satellites, comparing traditional methods like look-aside~\cite{Braun2019} among others. However, these methods suffer from high aggregated interference at the receiver~\cite{Jalali2024}. A significant amount of literature also explores spectrum sharing between terrestrial and satellite. Notably, a key technique which uses reverse pairing were proposed for flexible utilization of frequency bands to reduce the interference~\cite{Lee2024}.  In normal pairing, both systems use the same band for \cgls{dl} and another shared band for \cgls{ul}. By contrast, reverse pairing assigns one band to the \cgls{dl} of one system and the \cgls{ul} of the other, and vice-versa for the second band.
More recent studies shifted attention towards \cgls{leo}-\cgls{leo}, by providing various methods to reduce interference~\cite{Al-Hraishawi2023,Takahashi2019}. However, the existing interference mitigation techniques provide limited gain in high-interference scenarios; therefore, an effective strategy is desired, due to emerging dense \cgls{leo} constellations.

In this work, we propose a novel scheduling methodology for satellites operating on common bands, where the two fixed frequency bands can dynamically switch roles between \cgls{dl} and \cgls{ul}. This flexibility introduces an additional degree-of-freedom in scheduling, and improve the overall system performance. The selection of \cgls{dl} and \cgls{ul} frequency bands from the two available options is treated as a variable, which we refer to as the \textit{spin}, and the overall mechanism is termed \textit{spinning bands}. The idea of spin is inspired by~\cite{Popovski2015}, where authors optimized the transmission directions of \cgls{dl} and \cgls{ul} in a \cgls{tdd} system.

\begin{figure} 
    \centering
    {\footnotesize
    \begin{tikzpicture}[yscale=0.7, every path/.style={-latex, >=Latex}] 
          \coordinate (sat3) at (-3.1,2.2);
          \coordinate (sat4) at (-.8, 2.2);
          \coordinate (ue3)  at (-2.5,0.2);
          \coordinate (ue4)  at (-1.5,0.2);
          \coordinate (sat1) at (1.5,2.2);
          \coordinate (sat2) at (3.7, 2.2);
          \coordinate (ue1)  at (2.1,0.2);
          \coordinate (ue2)  at (3.1,0.2);     
         \pgfmathsetlengthmacro{\delta}{10pt}


        \node (s1) at (sat1) {\includegraphics[width=.6cm, angle=-100]{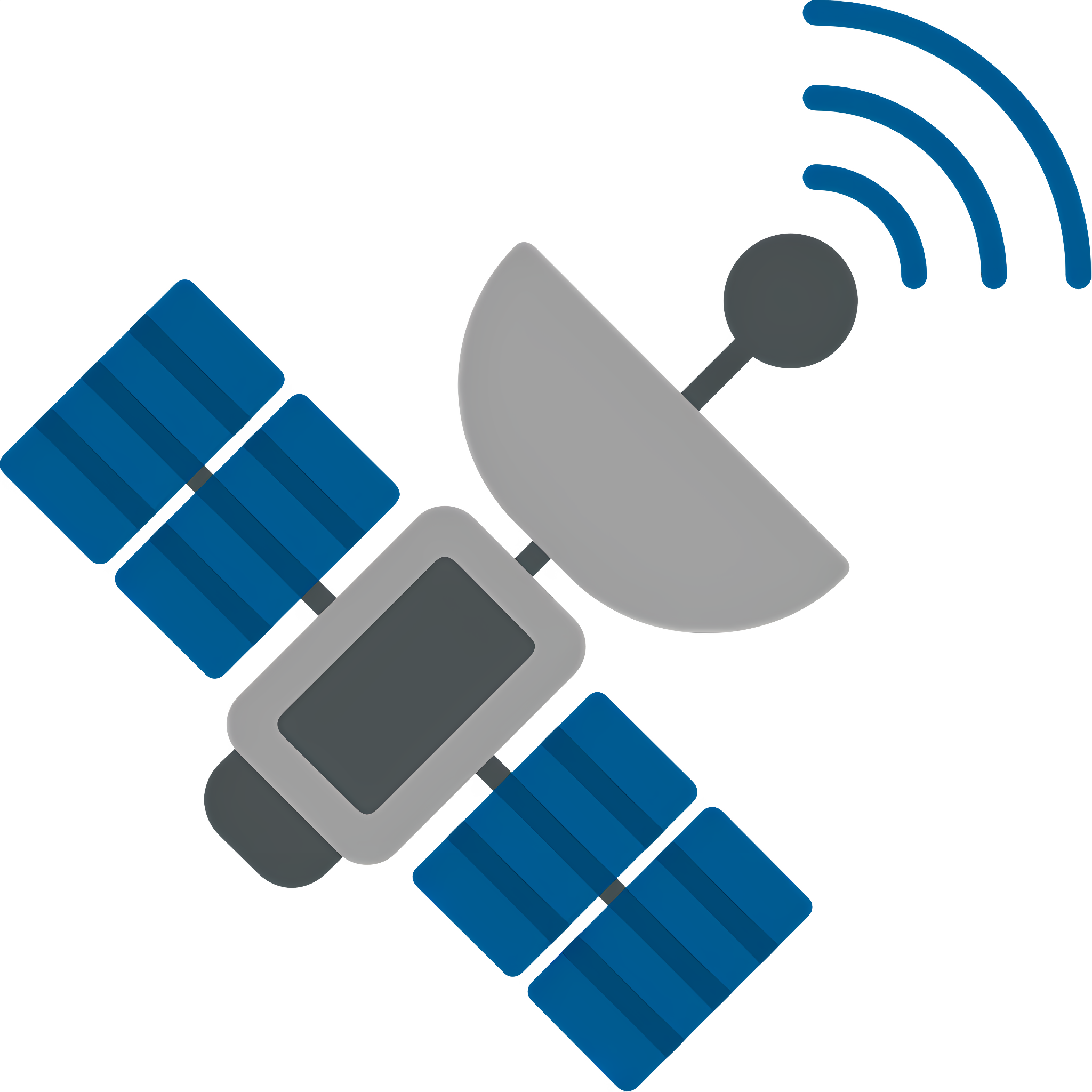}};
        \node (s2) at (sat2) {\includegraphics[width=.6cm, angle=-170]{Fig/sat.png}};
        \node (u1) at (ue1) {\includegraphics[width=0.7cm]{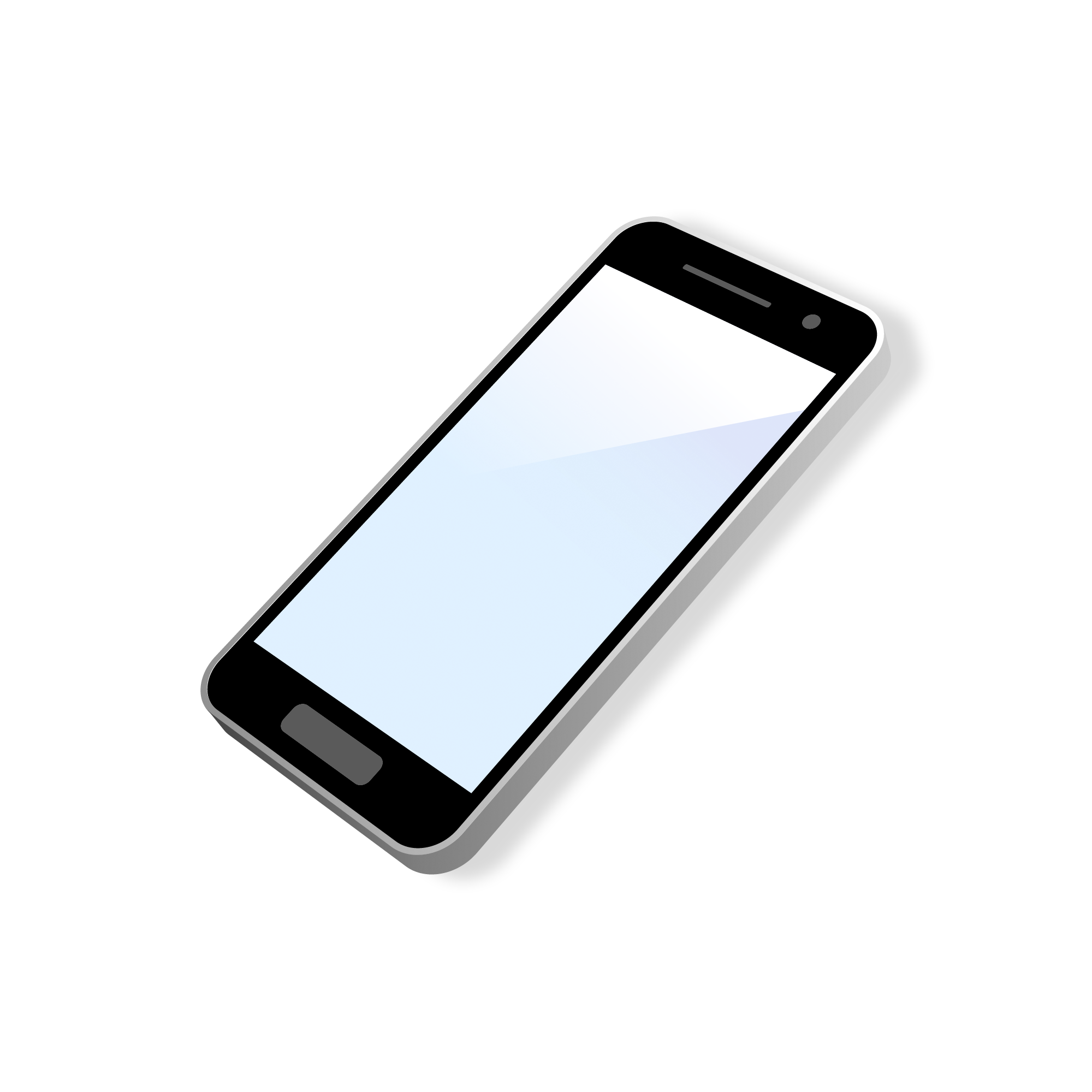}};
        \node (u2) at (ue2) {\includegraphics[width=0.7cm]{Fig/phn.png}};

        \coordinate (sat1_anchor_dl) at ($(sat1) + (0.1,-0.35)$);
        \coordinate (sat1_anchor_ul) at ($(sat1) + (0.2,-0.2)$);
        \coordinate (sat2_anchor_dl) at ($(sat2) + (-0.25,-0.2)$);
        \coordinate (sat2_anchor_ul) at ($(sat2) + (-0.1,-0.3)$);
        \coordinate (ue1_anchor_dl) at ($(ue1) + (0,0)$);
        \coordinate (ue1_anchor_ul) at ($(ue1) + (0.1,0.15)$);
        \coordinate (ue2_anchor_dl) at ($(ue2) + (-0.05,0.12)$);
        \coordinate (ue2_anchor_ul) at ($(ue2) + (0.1,0.05)$);

        \node[above=-0.35cm] at ($(ue1_anchor_dl)!0.5!(ue2_anchor_ul)$) {\includegraphics[width=.55cm, height=.9cm]{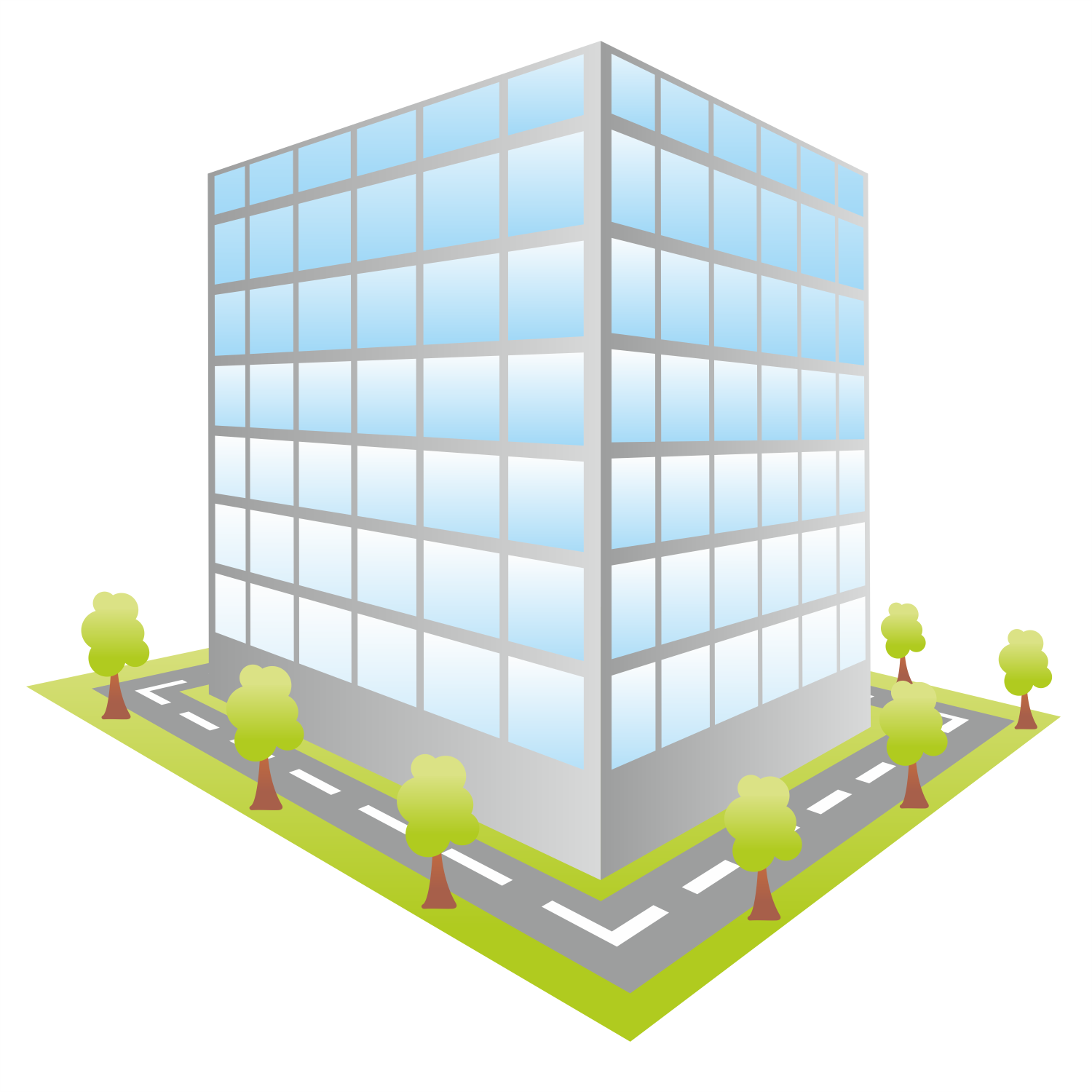}};
        \draw[->,color=myblue, thick, line cap=round]  (sat1_anchor_dl) -- (ue1_anchor_dl);
        \draw[->,color=red!70, thick, line cap=round]  (ue1_anchor_ul) -- (sat1_anchor_ul);
        \draw[->, color=red!70, thick, line cap=round]  (sat2_anchor_dl) -- (ue2_anchor_dl);
        \draw[->,color=myblue, thick, line cap=round]  (ue2_anchor_ul) -- (sat2_anchor_ul);       

        \draw[->,dashed, color=myblue!50, thick, line cap=round] (sat1_anchor_dl) -- (sat2_anchor_ul);
        \draw[->,dashed, color=red!40, thick, line cap=round] (sat2_anchor_dl) -- (sat1_anchor_ul);       
        \draw[->,dashed, color=myblue, thick, line cap=round] (ue2_anchor_ul) -- (ue1_anchor_dl);
        \draw[->,dashed, color=red!70, thick, line cap=round] (ue1_anchor_ul) -- (ue2_anchor_dl);     

        \node (s3) at (sat3) {\includegraphics[width=.6cm, angle=-100]{Fig/sat.png}};
        \node (s4) at (sat4) {\includegraphics[width=.6cm, angle=-170]{Fig/sat.png}};
        \node (u3) at (ue3) {\includegraphics[width=0.7cm]{Fig/phn.png}};
        \node (u4) at (ue4) {\includegraphics[width=0.7cm]{Fig/phn.png}};

        \coordinate (sat3_anchor_dl) at ($(sat3) + (0.1,-0.35)$);
        \coordinate (sat3_anchor_ul) at ($(sat3) + (0.2,-0.2)$);
        \coordinate (sat4_anchor_dl) at ($(sat4) + (-0.25,-0.2)$);
        \coordinate (sat4_anchor_ul) at ($(sat4) + (-0.1,-0.3)$);
        \coordinate (ue3_anchor_dl) at ($(ue3) + (-.05,0.15)$);
        \coordinate (ue3_anchor_ul) at ($(ue3) + (0.1,0.15)$);
        \coordinate (ue4_anchor_dl) at ($(ue4) + (-0.05,0.12)$);
        \coordinate (ue4_anchor_ul) at ($(ue4) + (0.1,0.18)$);

        \draw[->,color=myblue, thick, line cap=round]  (sat3_anchor_dl) -- (ue3_anchor_dl);
        \draw[->,color=red!70, thick, line cap=round]  (ue3_anchor_ul) -- (sat3_anchor_ul);
        \draw[->,color=myblue, thick, line cap=round]  (sat4_anchor_dl) -- (ue4_anchor_dl);
        \draw[->,color=red!70, thick, line cap=round]  (ue4_anchor_ul) -- (sat4_anchor_ul);       

        \draw[->, dashed, color=myblue, thick, line cap=round] (sat3_anchor_dl) -- (ue4_anchor_dl);
        \draw[->, dashed, color=myblue, thick, line cap=round] (sat4_anchor_dl) -- (ue3_anchor_dl);       
        \draw[ ->,dashed, color=red!70, thick, line cap=round] (ue4_anchor_ul) -- (sat3_anchor_ul);
        \draw[->, dashed, color=red!70, thick, line cap=round] (ue3_anchor_ul) -- (sat4_anchor_ul);     

        \draw[line width=0.5pt, line cap=round] ($(sat1_anchor_dl)!0.5!(sat2_anchor_ul)$) ellipse (.08cm and 0.27cm);
        \node[above=0.2cm] at ($(sat1_anchor_dl)!0.5!(sat2_anchor_ul)$) {\shortstack{low \\ interference}};

		\draw[line width=0.5pt, line cap=round] 
		($ (ue3_anchor_dl)!0.5!(ue4_anchor_ul) + (0,0.7) $) 
		ellipse (.08cm and 1cm);
		
		\node[above=.7cm] 
		at ($ (ue3_anchor_dl)!0.5!(ue4_anchor_ul) + (0,0.7) $) 
		{\shortstack{ high \\ interference}};
        \draw[-, thick, color=red!70, line cap=round] ($(ue2_anchor_dl)!0.5!(ue3_anchor_dl) + (-25pt, 10pt)$) -- ($(ue2_anchor_dl)!0.5!(ue3_anchor_dl) + (-5pt, 10pt)$) node[right] {band 1};
        \draw[-, thick, color=myblue, line cap=round] ($(ue2_anchor_dl)!0.5!(ue3_anchor_dl) + (-25pt, 1pt)$) -- ($(ue2_anchor_dl)!0.5!(ue3_anchor_dl) + (-5pt, 1pt)$) node[right] {band 2};

        \node[below] at ($(ue1)!0.5!(ue2) + (0, -\delta)$) {(b) with \cgls{dl}-\cgls{ul} band spin};
        \node[below] at ($(ue3)!0.5!(ue4) + (0, -\delta)$) {(a) without \cgls{dl}-\cgls{ul} band spin};
    \end{tikzpicture}
    }
  \caption{Motivation for dynamic downlink–uplink (DL–UL) band selection, shown for two satellites: (a) Fixed-band allocation, each satellite uses predetermined DL and UL bands, leading to strong interference. (b) Dynamic band selection, satellites flexibly assign bands for DL or UL, reducing inter-satellite interference through spatial separation and directive antennas.}
    \label{fig:1}
\end{figure}



To appreciate the advantages of dynamic \cgls{dl}-\cgls{ul} band selection in the satellite systems, consider an example illustrated in Fig.~\ref{fig:1}. In a conventional setup, Fig.~\ref{fig:1}(a), each satellite communicates with its own \cgls{ue} using fixed frequency bands for \cgls{dl} and \cgls{ul}. When these bands are reused across satellites, interference arises. Now consider the alternative in Fig.~\ref{fig:1}(b), where the satellites are granted the flexibility to dynamically assign the direction of transmission for each band, and assume the satellites use different bands for \cgls{dl}, correspondingly for \cgls{ul}. In this case, the dominant source of interference shifts to satellite-to-satellite and \cgls{ue}-to-\cgls{ue} links. Since, the satellites are generally separated by a large distances and equipped with highly directive antennas, this reduces the satellite-satellite interference. Therefore, the only source of interference is \cgls{ue}-\cgls{ue}. As the \cgls{ue}-\cgls{ue} channel depending on the environment maybe heavily attenuated, and also modern \cglspl{ue} are equipped with multiple patch antennas and capable of \cgls{ul} beamforming, thereby choosing opposite bands may become an essential tool for increasing the throughput. However, for more than two satellites, the decision becomes more challenging and even with high \cgls{ue}-\cgls{ue} interference, depending on the \cgls{ue} distributions, assigning different \cgls{dl} and \cgls{ul} bands to different systems may lead to increase in throughput. This motivates us to formulate a joint optimization problem that incorporates dynamic band assignment, user scheduling, and power allocation. 

 The rest of the paper is organized as follows: the system model, then the optimization problem involving spin, power, and scheduling is presented in Section~\ref{sec:2}. The solution approach to the optimization is discussed in Section~\ref{sec:3}, numerical results are presented in Section~\ref{sec:4}, and finally conclusions are drawn in  Section~\ref{sec:5}.

	\section{System Model \& Problem Formulation}
	\label{sec:2}

	Consider a group comprising $J$ \cgls{leo} satellites, each equipped with $N$ antennas, serving a region with a total of $K$ single-antenna \cglspl{ue}, as illustrated in Fig.~\ref{fig:2}. In this work, the problem formulation is done over a single snapshot, during which the satellite positions and resource allocations are assumed to remain fixed. The satellites and \cglspl{ue} are indexed by the sets $\mathcal{J} = \{1, \cdots, J\}$ and $\mathcal{K} = \{1, \cdots, K\}$, respectively. Two frequency bands are available for \cgls{dl} and \cgls{ul} communications, where band $l$, $l \in \{1, 2\}$, has center frequency $f_l$ and bandwidth $B_l$. Satellite $j \in \mathcal{J}$ serves a subset of the \cglspl{ue} on these bands. Each \cgls{ue} maintains separate connections for \cgls{dl} and \cgls{ul}, either to the same or to separate satellites.
	\cgls{ue} association to satellite~$j$ in the \cgls{dl} is indicated by a binary variable $d_{kj} \in \{0,1\}$, where $d_{kj} = 1$ means \cgls{ue} $k$ is served in \cgls{dl} by satellite~$j$, and $d_{kj} = 0$ otherwise. Similarly, \cgls{ul} association is indicated by a binary variable $u_{kj} \in \{0,1\}$, where $u_{kj} = 1$ means \cgls{ue} $k$ is served in \cgls{ul} by the satellite~$j$, and $u_{kj} = 0$ otherwise. The band assignment for satellite $j$ is represented by spin binary variable $r_j \in \{0,1\}$. If $r_j = 1$, band $l = 1$ is used for \cgls{dl} and band $l = 2$ for \cgls{ul}; if $r_j = 0$, the roles of the bands are reversed, as illustrated in Fig.~\ref{fig:3}. This is, essentially, a \cgls{fdd} system at the satellite $j$ with improved interference management capabilities through flexible band assignment. Still, each transceiver has to use orthogonal frequency bands for \cgls{ul} and \cgls{dl} transmissions. At \cgls{ue} $k$, this can be enforced by the condition $
 \sum_{j \in \mathcal{J}} d_{kj} r_{j} =  \sum_{j \in \mathcal{J}} u_{kj} r_{j}$. This, however, is only necessary if the \cgls{ue} is simultaneously served in \cgls{dl} and \cgls{ul}. Thus,
\begin{multline}
\sum_{j \in \mathcal{J}} d_{kj} = 1 \quad \mathrm{and}\quad \sum_{j \in \mathcal{J}} u_{kj} = 1
\\\implies
\sum_{j \in \mathcal{J}} d_{kj} r_{j} = \sum_{j \in \mathcal{J}}  u_{kj} r_{j}.
\label{eq:1}
\end{multline}
For all other cases, this condition is not relevant.

\begin{figure}[t]
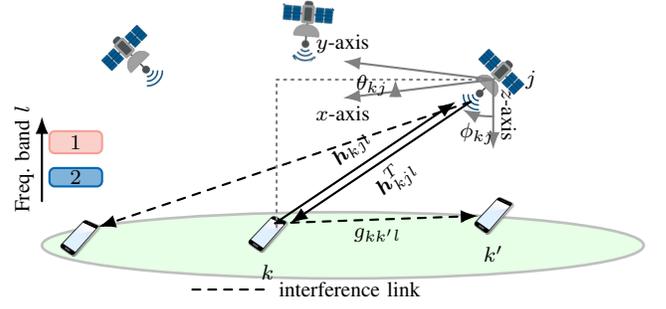

	\centering
	{\footnotesize
		\begin{tikzpicture}[every path/.style={>=Latex}, yscale=0.7] 
			\filldraw[fill=green!10, draw=gray!50, line width=1pt] (0,0) ellipse (4cm and .63cm);
			\draw[-,dashed, thick, line cap=round] (-2,-0.84) -- (-1,-0.84); 
			\node[anchor=west] at (-.95,-0.84) {interference link};
			
			\node (sat1) at (2,3.15) {\includegraphics[width=.8cm, angle=-180]{Fig/sat.png}}; 
			\node (sat2) at (-2.8,3.6) {\includegraphics[width=.8cm, angle=-90]{Fig/sat.png}};
			\node (sat3) at (-.5,4.2) {\includegraphics[width=.8cm, angle=-140]{Fig/sat.png}};
			\node at (sat1.east)  {$j$};
			
			\node (ue1) at (-1,0.21) {\includegraphics[width=0.8cm]{Fig/phn.png}}; 
			\node (ue2) at (2,0.56)   {\includegraphics[width=0.8cm]{Fig/phn.png}};
			\node (ue3) at (-3.5,0.14){\includegraphics[width=0.8cm]{Fig/phn.png}};
			\node at (ue1.south)  {$k$};
			\node at (ue2.south)  {$k'$};
			
			\draw[<-,thick, line cap=round] (1.5,2.73) -- (-0.9,0.42)  node[midway, above, sloped] {$\bm h_{kjl}$};
			\draw[->, thick, line cap=round] (1.7,2.73) -- (-0.7,0.42)  node[midway, below, sloped] {$\bm h_{kjl}^T$};    
			\draw[->, thick, dashed, line cap=round] (1.7,2.73) -- (-3.25,0.35);
			\draw[->, thick, dashed, line cap=round] (-0.9,0.42) -- (1.8,0.56) node[midway, below, sloped] {$g_{kk'l}$};
			
			\draw[->, thick, line cap=round] (-4,0.84) -- (-4,2.45) node[midway, sloped, above] {Freq. band $l$};
			\filldraw[fill=myred!50, draw=myred, thick, rounded corners=2pt] (-3.9,1.75) rectangle (-3.2,2.17);
			\node at (-3.55,1.96) {\textcolor{black}{$1$}};
			\filldraw[fill=myblue!50, draw=myblue, thick, rounded corners=2pt] (-3.9,1.12) rectangle (-3.2,1.47);
			\node at (-3.55,1.295) {\textcolor{black}{$2$}};
			
			\draw[-, thick, dotted, draw=gray, line cap=round] (-0.88,0.35) -- (-.88, 3.15); 
			\draw[-, thick, dotted, draw=gray, line cap=round] (-.88,3.15) -- (2, 3.15); 
			\draw[->, thick, draw=gray, line cap=round]  (2, 3.15) -- (0, 2.8) node[left, sloped, below] {$x$-axis}; 
			\draw[->, thick, draw=gray, line cap=round]  (2, 3.15) -- (0, 3.5) node [left, sloped, above] {$y$-axis}; 
			\draw[->, thick, draw=gray, line cap=round]  (2, 3.15) -- (2, 1.82) node [midway, sloped, above] {$z$-axis}; 
			\draw[<-, thick, draw=gray, line cap=round] (0.7, 3.15) arc[start angle=180, end angle=198, radius=0.7cm] node[midway, left] {$\theta_{kj}$};
			\draw[->, thick, draw=gray, line cap=round] (2, 2.45) arc[start angle=270, end angle=234, radius=0.7cm] node[midway, below] {$\phi_{kj}$};
			
		\end{tikzpicture}
	}
	\caption{A system of $J$ \cgls{leo} satellites in orbit serving $K$ \cglspl{ue} over region.}
	\label{fig:2}
\end{figure}
\begin{figure}[t]
	\centering
	{\footnotesize
		\begin{tikzpicture}[yscale=0.5, every path/.style={>=Latex}] 
						
			\draw[->, thick] (1,0) -- (6.5,0) node[right]{Frequency};
			
			\filldraw[fill=myblue!50, draw=myblue, thick, rounded corners=2pt] (2,0.2) rectangle (3.5,0.7) node[midway]{{UL}};
			\filldraw[fill=myred!50, draw=myred, thick, rounded corners=2pt] (4,0.2) rectangle (5.5,0.7) node[midway]{{DL}};
			
			\filldraw[fill=myblue!50, draw=myblue, thick, rounded corners=2pt] (2,0.9) rectangle (3.5,1.4) node[midway]{{DL}};
			\filldraw[fill=myred!50, draw=myred, thick, rounded corners=2pt] (4,0.9) rectangle (5.5,1.4) node[midway]{{UL}};
			
			\node at (1,0.45) {$r_j = 1$};
			\node at (1,1.15) {$r_j = 0$};
			
			\draw[thick] (2.75,0) -- (2.75,0.2);
			\draw[thick] (4.75,0) -- (4.75,0.2);
			
			\node[below] at (2.75,0) {$f_2$};
			\node[below] at (4.75,0) {$f_1$};
			
		\end{tikzpicture}
	}
	\caption{Introduced spin variable for a satellite $j$.}
	\label{fig:3}
\end{figure}
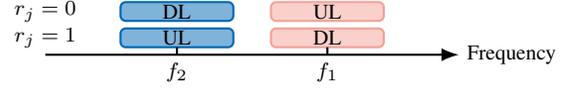

	We denote the band-$l$ channel from satellite $j$ to \cgls{ue}~$k$ as $\bm h_{kjl} \in \mathbb{C}^{N \times 1}$. This channel has a dominant \cgls{los} component with only small non-\cgls{los} effects in the ground-segment. The \cglspl{ue} are assumed to know their own position and those of all satellites. In particular, let $m_{kj}$, $\theta_{kj}$, and $\phi_{kj}$ be the distance, azimuth, and elevation angle from satellite $j$ to \cgls{ue} $k$. Then, the \cgls{los} component of $\bm h_{kjl}$ is $\sqrt{\beta_{kjl}}\, \bm b_l(\theta_{kj}, \phi_{kj})$ with path loss $\beta_{kjl} = ( {c}/{4 \pi m_{kj} f_l})^2$ and array-response vector as
	\begin{equation}
		\bm b_l = 
		\Big(\Big[ \exp\Big(-\frac{j 2\pi f_l}{c}( D^x_n \psi^x + D^y_n \psi^y)\Big) \Big]_{n=1}^N\Big)^T,
		\label{eq:2}
	\end{equation}
	where $\psi^x = \cos{( \phi_{kj} )}\cos{(\theta_{{kj}})}, \psi^y = \cos{( \phi_{kj} )}\sin{(\theta_{kj})}$, and $
	D^x_n, D_n^y$ denote the $x$ and $y$-coordinate, respectively of $n$-th antenna element. We model $\bm h_{kjl}$ as purely \cgls{los} but note that statistical \cgls{csi} is straightforward to incorporate into $\beta_{kjl}$. Further, we assume negligible interference between satellites due to the high directivity and downward positioning of satellite antenna array. However, we need to take inter-\cgls{ue} interference into account. Since accurate and instantaneous inter-\cgls{ue} \cgls{csi} is infeasible to obtain, we model the band-$l$ channel $g_{kk'l} = g_{k'kl}$ from \cgls{ue} $k$ to $k'$ based on position-dependent long-term statistics.

Now, consider that \cgls{ue}~$k$ is served by satellite~$j$ in the \cgls{dl}, indicated by the binary variable $d_{kj} = 1$. The effective \cgls{dl} interference channel from satellite~$j' \in \mathcal{J}$ to $k$, while serving another \cgls{ue}~$k' \in \mathcal{K}_{-k}$ in the \cgls{dl} (i.e., $d_{k'j'}=1$), is given by
\begin{align}
   (\bm \gamma^{\mathrm{dl}}_{kjj'} )^T= \overline{\lvert r_j - r_{j'} \rvert} \,  [ r_{j'} \, \bm h_{kj',l=1}^T  + \overline{r_{j'}} \, \bm h_{kj',l=2}^T ], 
   \label{eq:3}
\end{align}
where $r_j$ and $r_{j'}$ denote the spins of satellites~$j$ and~$j'$, and $\lvert r_j - r_{j'} \rvert$ captures the relative spin between them. This interference arises when both satellites operate with the same spin.  Here, $\overline{a}$ denotes the binary complement of $a \in \{0,1\}$, and for a set $\mathcal{A}$, $\mathcal{A}_{-i} = \mathcal{A} \setminus \{i\}$ for $i \in \mathcal{A}$. Note that $\bm \gamma^{\mathrm{dl}}_{kjj'}$ also represents the direct \cgls{dl} channel between satellite~$j$ and \cgls{ue}~$k$ when $j' = j$, since $\overline{\lvert r_j - r_j\rvert} = 1$. Similarly, if \cgls{ue}~$k$ is served by satellite~$j$ in the \cgls{ul}, indicated by $u_{kj} = 1$, then the effective \cgls{ul} interference channel from \cgls{ue}~$k' \in \mathcal{K}$ to satellite~$j$, where $k'$ is served by satellite~$j' \in \mathcal{J}$ in the \cgls{ul} (i.e., $u_{k'j'} = 1$), is
\begin{align}
   \bm \gamma^{\mathrm{ul}}_{k'jj'} = \overline{\lvert r_j - r_{j'} \rvert} \, [ r_{j'} \, \bm h_{k'j,l=2} + \overline{r_{j'}} \, \bm h_{k'j,l=1} ].
   \label{eq:4}
\end{align}
Note that, $\bm \gamma^{\mathrm{ul}}_{k'jj'}$ also captures the direct \cgls{ul} channel from \cgls{ue}~$k$ to satellite~$j$ when $k' = k$ and $j' = j$; and this interference arises when both satellites operate with the same spin. 
Additionally, a \cgls{ue} may cause interference to another \cgls{ue} during \cgls{ul} transmission. Suppose \cgls{ue}~$k' \in \mathcal{K}_{-k}$ is served in the \cgls{ul} by satellite~$j' \in \mathcal{J}_{-j}$ (i.e., $u_{k'j'} = 1$), while \cgls{ue}~$k$ is served in the \cgls{dl} by satellite~$j$ (i.e., $d_{kj} = 1$). Then, the effective interference channel from \cgls{ue}~$k'$ to \cgls{ue}~$k$ is given by
\begin{align}
   \nu_{kk'jj'} = \lvert r_j - r_{j'} \rvert \, [ r_{j'} \, g_{k'k,l=1} + \overline{r_{j'}} \, g_{k'k,l=2} ],
   \label{eq:5}
\end{align}
where $\nu_{kk'jj'} = \nu_{k'kjj'}$ due to reciprocity. This interference only arises if the spins of the associated satellites $j$ and $j'$ are opposite.

	The received signals at \cgls{ue}~$k$ in the \cgls{dl} and at satellite~$j$ in the \cgls{ul}, corresponding to link~$kj$, are denoted by $y_{kj}^{\mathrm{dl}}$ and $y_{kj}^{\mathrm{ul}}$, respectively, and can be expressed as
    \begin{subequations}\label{eq:6}
    \begin{align}
		y_{kj}^{\text{dl}} =&\, d_{kj} \, \sqrt{p_{k}^{\text{dl}}} \, \lVert \bm \gamma^{\mathrm{dl}}_{kjj} \rVert \, s^{\text{dl}}_{k}  \nonumber\\
		&+ \sum_{j' \in \mathcal{J}} \sum_{k' \in \mathcal{K}_{-k}} d_{k'j'} \, \sqrt{p_{k'}^{\text{dl}}} \, ( \bm \gamma^{\mathrm{dl}}_{kjj'} )^T  \bm w_{k'j'} \, s^{\text{dl}}_{k'} \nonumber\\
		&+ \sum_{j' \in \mathcal{J}_{-j}} \sum_{k' \in \mathcal{K}_{-k}} u_{k'j'} \, \sqrt{p_{k'}^{\text{ul}}} \, \nu_{kk'jj'} \, s^{\text{ul}}_{k'} + n^{\mathrm{dl}}_{kj},
		\label{eq:6a} \\
		y^{\text{ul}}_{kj} =&\, u_{kj} \sqrt{p_{k}^{\text{ul}}} \, \lVert \bm \gamma^{\mathrm{ul}}_{kjj} \rVert \, s^{\text{ul}}_{k} \nonumber \\
		&+ \sum_{j' \in \mathcal{J}} \sum_{k' \in \mathcal{K}_{-k}} u_{k'j'} \, \sqrt{p_{k'}^{\text{ul}}} \, \bm v_{kj}^T \bm \gamma^{\mathrm{ul}}_{k'jj'} \, s^{\text{ul}}_{k'} + n^{\mathrm{ul}}_{kj},
		\label{eq:6b}
	\end{align}
    \end{subequations}
	where $\bm w_{k'j'} \in \mathbb{C}^{N \times 1}$ denotes the \cgls{dl} \gls{mrt} precoder for link $k'j'$, and the vector $\bm v_{kj} \in \mathbb{C}^{N \times 1}$ is the \gls{mrc} receiver at satellite~$j$ for \cgls{ue}~$k$ in the \cgls{ul}.
	Additionally, $s_{k'}^{\mathrm{dl}}$ denotes the unit-energy symbol transmitted in the \cgls{dl} towards \cgls{ue}~$k'$, whereas $s_{k'}^{\mathrm{ul}}$ represents the unit-energy symbol transmitted by \cgls{ue}~$k'$ in the \cgls{ul}. The receiver noise terms $n_{kj}^{\mathrm{dl}}$ for the \cgls{dl} and $n_{kj}^{\mathrm{ul}}$ for the \cgls{ul} are modeled as complex Gaussian with mean $0$ and variance $\sigma^2$. The corresponding transmit powers are denoted by $p_{k'}^{\mathrm{dl}}$ for the \cgls{dl}, i.e., from satellite to \cgls{ue}~$k'$, and by $p_{k'}^{\mathrm{ul}}$ for the \cgls{ul}, i.e., from \cgls{ue}~$k'$ to the satellite.

	\subsection{Formulation of the Optimization Problem}
	\label{sec:3a}
    
		\begin{figure*}[!b]
		\hrulefill 
		\vspace{0.1cm}
		\begin{minipage}{\textwidth}
			\begin{align}
				f_0&(\bm d, \bm r, \bm u, \bm p^{\text{dl}}, \bm p^{\text{ul}}) =   \sum_{k \in \mathcal{K}} \sum_{j \in\mathcal{J}} \Big[ \log_2 (1 +  u_{kj} \, p_{k}^{\text{ul}}  \lVert \bm \gamma^{\mathrm{ul}}_{kjj} \rVert^2 	 /  ( \sigma^2 +  \sum_{j' \in \mathcal{J}} \sum_{k' \in \mathcal{K}_{-k}} u_{k'j'} \, p_{k'}^{\text{ul}} \, \lvert \bm v_{kj}^T \bm \gamma^{\mathrm{ul}}_{k'jj'}\rvert^2   ) )	\nonumber\\
				&+ \log_2 ( 1 + d_{kj}\, p_{k}^{\text{dl}} \, \lVert \bm \gamma^{\mathrm{dl}}_{kjj}  \rVert^2/ ( \sigma^2 +  \sum_{j' \in \mathcal{J}} \sum_{k' \in \mathcal{K}_{-k}} d_{k'j'}\, p_{k'}^{\text{dl}} \, \lvert ( \bm \gamma^{\mathrm{dl}}_{kjj'} )^T  \bm w_{k'j'}\rvert^2    + \sum_{j' \in \mathcal{J}_{-j} }\sum_{k' \in \mathcal{K}_{-k}}  u_{k'j'} \, p_{k'}^{\text{ul}}\, \lvert \nu_{kk'jj'}\rvert^2) ) \Big]. 	\label{eq:7}
			\end{align} 
		\end{minipage}
	\end{figure*}
    
	We formulate a joint optimization problem considering user association, power control, and spins. The objective is the sum of both \cgls{dl} and \cgls{ul} spectral efficiency for all possible user-satellite pairs, and the mathematical program is written as
	\begin{subequations} \label{eq:8} %
		\begin{align}
			(\mathrm{P0): }&\underset{\bm d, \bm r, \bm u, \bm p^{\text{dl}}, \bm p^{\text{ul}}}{\text{max}}	\, \, \, \, f_0(\bm d, \bm r, \bm u, \bm p^{\text{dl}}, \bm p^{\text{ul}}) \nonumber\\
			\text{s.t. \hspace{1pt}	}
			& \text{FDD condition in }\eqref{eq:1}, \text{ for all } k  \label{eq:8b} \\
			& \sum_{j \in \mathcal{J}} d_{kj} \leq 1, \sum_{j \in \mathcal{J}} u_{kj} \leq 1, \text{ for all } k\label{eq:8c}\\
			& \sum_{k \in \mathcal{K}} d_{kj} \,p_{k}^{\text{dl}}\leq p_{j}^{\text{max}}, \text{ for all } j\label{eq:8d} \\
			& p_{k}^{\text{ul}} \leq p_{k}^{\text{max}}, \text{ for all } k\label{eq:8e} \\
			& p_{k}^{\text{dl}} \geq 0, \, p_{k}^{\text{ul}} \geq 0, \text{ for all } k \label{eq:8f}\\
                      & d_{kj}, r_{j}, u_{kj} \in \{0,1\}, \text{ for all } k \text{ and } j. \label{eq:8g}
		\end{align}
	\end{subequations}
    Here, $\bm d = [d_{kj}]_{k=1, j=1}^{k= K, j = J}, \bm u = [u_{kj}]_{k=1, j=1}^{k= K, j = J}, \bm r = [r_j]_{j=1}^J, \bm p^{\text{dl}} = [p_{k}^{\text{dl}}]_{k=1}^K$, and $\bm p^{\text{ul}} = [p_{k}^{\text{ul}}]_{k=1}^K$. The objective $f_0$ is defined in \eqref{eq:7}. Constraint~\eqref{eq:8c} restricts each \cgls{ue} to connect with at most one satellite in \cgls{dl} and at most one in \cgls{ul}.  
    Constraint~\eqref{eq:8d} limits the total \cgls{dl} transmit power of each satellite~$j$ to its maximum budget $p_j^{\text{max}}$, while \eqref{eq:8e} ensures that each \cgls{ue}~$k$ does not exceed its own \cgls{ul} transmit power budget $p_k^{\text{max}}$.  
    Constraint~\eqref{eq:8f} enforces non-negativity of all power variables, and \eqref{eq:8g} implies that the scheduling and spin variables are binary.

\section{Solution to $(\mathrm{P0})$}
\label{sec:3}

    \begin{figure*}[b]
        \hrulefill 
		\vspace{0.1cm}
		\begin{minipage}{\textwidth}
        \begin{align}
            	f_1&(\bm d, \bm r, \bm u, \bm p^{\text{dl}}, \bm p^{\text{ul}}, \,\bm \chi^{\text{dl}} , \bm \chi^{\text{ul}}  ) = \nonumber \\
				&  \sum_{k \in \mathcal{K}} \sum_{j\in \mathcal{J}}  \Big[ \log_2 (1+  	\chi^{\text{dl}}_{kj} )	+ \log_2 ( 1 + \chi^{\text{ul}}_{kj} ) - \chi^{\text{dl}}_{kj}  - \chi^{\text{ul}}_{kj}  + (1+  \chi^{\text{ul}}_{kj} ) u_{kj} \, p_{k}^{\text{ul}}  \lVert \bm \gamma^{\text{ul}}_{kjj}  \rVert^2 	 /  ( \sigma^2 +  \sum_{j' \in \mathcal{J}} \sum_{k' \in \mathcal{K}} u_{k'j'} \, p_{k'}^{\text{ul}}  \lvert \bm v^T_{kj} \bm \gamma^{\text{ul}}_{k'jj'} \rvert^2   ) \nonumber \\
				&+ (1+  \chi^{\text{dl}}_{kj} ) d_{kj\, } p_{k}^{\text{dl}} \, \lVert\bm \gamma^{\text{dl}}_{kjj} \rVert^2 / ( \sigma^2 +  \sum_{j' \in \mathcal{J}} \sum_{k' \in \mathcal{K}} d_{k'j'\, } p_{k'}^{\text{dl}} \, \lvert (\bm \gamma^{\text{dl}}_{kjj'})^T \bm w_{k'j'}\rvert^2    + \sum_{j' \in \mathcal{J}_{-j}} \sum_{k' \in \mathcal{K}_{-k} }  u_{k'j'} \, p_{k'}^{\text{ul}} \,\lvert \nu_{kk'jj'}\rvert^2) \Big]. \label{eq:9} \\
f_2&(\bm d, \bm r, \bm u, \bm p^{\text{dl}}, \bm p^{\text{ul}}, \,\bm \chi^{\text{dl}} , \bm \chi^{\text{ul}}, \bm \xi^{\text{dl}}, \bm \xi^{\text{ul}} ) =   \sum_{k \in \mathcal{K}} \sum_{j\in \mathcal{J}}  \Big[ \log_2 (1+  	\chi^{\text{dl}}_{kj} )	+ \log_2 ( 1 + \chi^{\text{ul}}_{kj} ) - \chi^{\text{dl}}_{kj}  - \chi^{\text{ul}}_{kj} 
    +  2 \xi_{kj}^{\text{dl}} \, d_{kj }\sqrt{(1+  	\chi^{\text{dl}}_{kj} ) p_{k}^{\text{dl}}} \, \lVert\bm \gamma^{\text{dl}}_{kjj}\rVert \nonumber \\
    &+ 2 \xi_{kj}^{\text{ul}} \, u_{kj }\sqrt{(1+  	\chi^{\text{ul}}_{kj} ) p_{k}^{\text{ul}}} \, \lVert\bm \gamma^{\text{ul}}_{kjj}  \rVert - (\xi_{kj}^{\text{dl}})^2 ( \sigma^2 +  \sum_{j' \in \mathcal{J}} \sum_{k' \in \mathcal{K}} d_{k'j'\, } p_{k'}^{\text{dl}} \, \lvert (\bm \gamma^{\text{dl}}_{kjj'})^T \bm w_{k'j'} \rvert^2    + \sum_{j' \in \mathcal{J}_{-j}} \sum_{k' \in \mathcal{K}_{-k}}  u_{k'j'} \, p_{k'}^{\text{ul}} \,\lvert \nu_{kk'jj'} \rvert^2) \nonumber \\ 
    &- (\xi_{kj}^{\text{ul}})^2 ( \sigma^2 +  \sum_{j' \in \mathcal{J}} \sum_{k' \in \mathcal{K}} u_{k'j'} \, p_{k'}^{\text{ul}}  \lvert \bm v^T_{kj} \bm \gamma^{\text{ul}}_{k'jj'} \rvert^2   )  \Big]. \label{eq:10}
        \end{align}
        \end{minipage}
    \end{figure*}
    
The mathematical program in \eqref{eq:8} is non-convex due to sum of logarithmic functions of \cglspl{sinr}, and the \cgls{sinr} is expressed as ratio of sums containing product of continuous power and binary variables. To solve, we do an equivalent transform, first using Lagrangian dual transform~\cite{Shen2018}, to take out the \cgls{sinr} outside logarithm; then, we apply Quadratic transform~\cite{Shen2018}, to make the ratios into a sum of two terms. After transformation, we do alternating optimization to solve the transformed problem, where we update block of variables. Finally, the algorithm is described.

First, we apply Lagrangian dual transform from Theorem~3 in \cite{Shen2018} to transform the problem in $(\mathrm{P0})$ as
\begin{subequations}\label{eq:18}
	\begin{align}
    (\mathrm{P1): }\underset{
      \substack{
        \bm d, \bm r, \bm u, \bm p^{\text{dl}}, \bm p^{\text{ul}},
        \bm \chi^{\text{dl}}, \bm \chi^{\text{ul}}
      }
    }{\text{max}}	\, \, \, \, &f_1(\bm d, \bm r, \bm u, \bm p^{\text{dl}}, \bm p^{\text{ul}}, \bm \chi^{\text{dl}} , \bm \chi^{\text{ul}}  ) \nonumber \\
		\text{s.t. \hspace{10pt}	}
          & \eqref{eq:8b}-\eqref{eq:8g}, \quad \chi_{kj}^{\mathrm{dl}}, \chi_{kj}^{\mathrm{ul}} \in \mathbb{R}_+
	\end{align}
\end{subequations}
where with the auxiliary variables $\bm \chi^{\text{dl}} = [\chi^{\text{dl}}_{kj}]_{k=1, j=1}^{k= K, j = J}$ and $\bm \chi^{\text{ul}} = [\chi^{\text{ul}}_{kj}]_{k=1, j=1}^{k= K, j = J}$, the objective $f_1$ is shown at the bottom in \eqref{eq:9}. Then, we apply the Quadratic transform from Theorem~1 in~\cite{Shen2018} to transform $(\mathrm{P1})$ as
\begin{subequations}\label{eq:19}
	\begin{align}
    	(\mathrm{P2): } \underset{ \substack{\bm d, \bm r, \bm u, \bm p^{\text{dl}}, \bm p^{\text{ul}}, \\ \bm \chi^{\text{dl}} , \bm \chi^{\text{ul}}, \bm \xi^{\text{dl}}, \bm \xi^{\text{ul}} }}{\text{max}}	\, \, \, \, &f_2(\bm d, \bm r, \bm u, \bm p^{\text{dl}}, \bm p^{\text{ul}}, \bm \chi^{\text{dl}} , \bm \chi^{\text{ul}}, \bm \xi^{\text{dl}}, \bm \xi^{\text{ul}} ) \nonumber\\
		\text{s.t. \hspace{5pt}	}
          & \eqref{eq:8b}-\eqref{eq:8g}, \\
          & \chi_{kj}^{\mathrm{dl}}, \chi_{kj}^{\mathrm{ul}} \in \mathbb{R}_+, \xi_{kj}^{\mathrm{dl}}, \xi_{kj}^{\mathrm{ul}} \in \mathbb{R}
	\end{align}
\end{subequations}
where with the auxiliary variables $\bm \xi^{\text{dl}} = [\xi^{\text{dl}}_{kj}]_{k=1, j=1}^{k= K, j = J}$ and $\bm \xi^{\text{ul}} = [\xi^{\text{ul}}_{kj}]_{k=1, j=1}^{k= K, j = J}$, the objective $f_2$ is shown at the bottom in \eqref{eq:10}. The equivalence between the transformed problems is established below:
\begin{proposition}

Consider a solution to $(\mathrm{P2})$ given by 
$(\bm d^\dagger, \bm r^\dagger, \bm u^\dagger, \bm p^{\mathrm{dl}\dagger}, \bm p^{\mathrm{ul}\dagger}, \bm \chi^{\mathrm{dl}\dagger}, \bm \chi^{\mathrm{ul}\dagger}, \bm \xi^{\mathrm{dl}\dagger}, \bm \xi^{\mathrm{ul}\dagger})$.
Then, the tuple $(\bm d^\dagger, \bm r^\dagger, \bm u^\dagger, \bm p^{\mathrm{dl}\dagger}, \bm p^{\mathrm{ul}\dagger}, \bm \chi^{\mathrm{dl}\dagger}, \bm \chi^{\mathrm{ul}\dagger})$ is also a solution to $(\mathrm{P1})$, and further, $(\bm d^\dagger, \bm r^\dagger, \bm u^\dagger, \bm p^{\mathrm{dl}\dagger}, \bm p^{\mathrm{ul}\dagger})$ is also a solution to $(\mathrm{P0})$.

\end{proposition}
\begin{proof}
 See Appendix~\ref{appendix:1}.
\end{proof}

To solve $(\mathrm{P0})$, we first solve $(\mathrm{P2})$. To solve $(\mathrm{P2})$, we update variables in three blocks: (i) update $\bm \xi^{\mathrm{dl}}$ and $\bm \xi^{\mathrm{ul}}$, (ii) update $\bm \chi^{\mathrm{dl}}$ and $\bm \chi^{\mathrm{ul}}$, and (iii) jointly update $\bm d, \bm u, \bm p^{\mathrm{dl}}$, and $\bm p^{\mathrm{ul}}$. Since the spin variables $\bm r$ add extra complexity, we perform an exhaustive search over all possible spin configurations.

\subsection{Update $\bm \chi^{\mathrm{dl}}$ and $\bm \chi^{\mathrm{ul}}$}
\label{sec:3a}
The objective $f_1$ is strictly concave in $\chi^{\text{dl}}_{kj}$ and $\chi^{\text{ul}}_{kj}$. Hence, by keeping other parameters fixed, the unique optimal value of $\chi^{\text{dl}}_{kj}$ and $\chi^{\text{ul}}_{kj}$ can be obtained from $\partial f_1 / \partial \chi^{\text{dl}}_{kj} = 0$ and $\partial f_1 / \partial \chi^{\text{ul}}_{kj} = 0$, respectively, as
\begin{subequations}\label{eq:14}
    \begin{align}
    \chi_{kj}^{\text{dl}\star} &=\, d_{kj}\, p_{k}^{\text{dl}} \, \lVert \bm \gamma^{\mathrm{dl}}_{kjj} \rVert^2 \nonumber \\ \Big/ &\Big( \sigma^2 +  \sum_{j' \in \mathcal{J}} \sum_{k' \in \mathcal{K}_{-k}} d_{k'j'}\, p_{k'}^{\text{dl}} \, \lvert ( \bm \gamma^{\mathrm{dl}}_{kjj'} )^T  \bm w_{k'j'}\rvert^2  \nonumber \\
   &+ \sum_{j' \in \mathcal{J}_{-j} }\sum_{k' \in \mathcal{K}_{-k}}  u_{k'j'} \, p_{k'}^{\text{ul}}\, \lvert \nu_{kk'jj'}\rvert^2 \Big),\\
    \chi_{kj}^{\text{ul}\star} &= u_{kj} \, p_{k}^{\text{ul}}  \lVert \bm \gamma^{\mathrm{ul}}_{kjj} \rVert^2 	\nonumber\\ \Big/ & \Big( \sigma^2 +  \sum_{j' \in \mathcal{J}} \sum_{k' \in \mathcal{K}_{-k}} u_{k'j'} \, p_{k'}^{\text{ul}} \, \lvert \bm v_{kj}^T \bm \gamma^{\mathrm{ul}}_{k'jj'}\rvert^2 \Big).
\end{align}
\end{subequations}
Here, $\chi^{\text{dl}}_{kj}$ and $\chi^{\text{ul}}_{kj}$ are updated simultaneously using the above equations, since they are independent of each other.

\subsection{Update $\bm \xi^{\mathrm{dl}}$ and $\bm \xi^{\mathrm{ul}}$}
The objective $f_2$ is strictly concave in $\xi^{\text{dl}}_{kj}$ and $\xi^{\text{ul}}_{kj}$. Hence, by keeping other parameters fixed, the unique optimal value of $\xi^{\text{dl}}_{kj}$ and $\xi^{\text{ul}}_{kj}$ can be obtained from $\partial f_2 / \partial \xi^{\text{dl}}_{kj} = 0$ and $\partial f_2 / \partial \xi^{\text{ul}}_{kj} = 0$, respectively, as
\begin{subequations}\label{eq:13}
    \begin{align}
        \xi^{\text{dl}\star}_{kj} = \,&d_{kj }\sqrt{(1+  \chi^{\text{dl}}_{kj}) p_{k}^{\text{dl}}} \, \lVert\bm \gamma^{\text{dl}}_{kjj}  \rVert  \nonumber \\
         \Big/ &\Big( \sigma^2  +\sum_{j' \in \mathcal{J}} \sum_{k' \in \mathcal{K}} d_{k'j'}\, p_{k'}^{\text{dl}} \, \lvert (\bm \gamma^{\text{dl}}_{kjj'})^T \bm w_{k'j'} \rvert^2    \nonumber \\
        & + \sum_{j' \in \mathcal{J}_{-j}} \sum_{k' \in \mathcal{K}_{-k}}  u_{k'j'} \, p_{k'}^{\text{ul}} \lvert \, \nu_{kk'jj'}\rvert^2 \Big), \\
         \xi^{\text{ul}\star}_{kj}  = \,&u_{kj }\sqrt{(1+  	\chi^{\text{ul}}_{kj}) p_{k}^{\text{ul}}} \, \lVert\bm \gamma^{\text{ul}}_{kjj} \rVert  \nonumber\\
            &\Big/\Big( \sigma^2  + \sum_{j' \in \mathcal{J}} \sum_{k' \in \mathcal{K}} u_{k'j'} \, p_{k'}^{\text{ul}}   \lvert \bm v^T_{kj} \bm \gamma^{\text{ul}}_{k'jj'} \rvert^2  \Big).
    \end{align}
\end{subequations}
Similar to \ref{sec:3a}, $\xi^{\text{dl}}_{kj}$ and $\xi^{\text{ul}}_{kj}$ are updated simultaneously.

\subsection{Update Jointly  $\bm d,  \bm u, \bm p^{\text{dl}},\bm p^{\text{ul}}$}
We jointly update $(\bm d,  \bm u, \bm p^{\text{dl}}, \bm p^{\text{ul}})$ as
 \begin{subequations}\label{eq:22}
 		\begin{align}
		&(\mathrm{P3): } \underset{ \substack{\bm d,  \bm u, \bm p^{\text{dl}}, \bm p^{\text{ul}}}}{\text{max}}  f_2( \bm d, \bm r, \bm u, \bm p^{\text{dl}}, \bm p^{\text{ul}}, \bm \chi^{\text{dl}\star} , \bm \chi^{\text{ul}\star}, \bm \xi^{\text{dl}\star},  \bm \xi^{\text{ul}\star} ) \nonumber\\
		& \qquad \hspace{50pt} \text{s.t. } 
		\eqref{eq:8b}-\eqref{eq:8g}.
	\end{align}
\end{subequations}
The above mathematical program involves the product of the square root of the power variable and a binary variable, as well as the product of the power variable and a binary variable. To handle these nonlinear terms, we apply the following substitutions $
	t_{k}^{\mathrm{dl}} = \sqrt{p_{k}^{\mathrm{dl}}}, t_{k}^{\mathrm{ul}} = \sqrt{p_{k}^{\mathrm{ul}}}, \text{for } k = 1,\dots,K,$
and define $	\bm t^{\mathrm{dl}} = [t_{k}^{\mathrm{dl}}]_{k=1}^K, \bm t^{\mathrm{ul}} = [t_{k}^{\mathrm{ul}}]_{k=1}^K$. This substitution is one-to-one for $t_{k}^{\mathrm{dl}} \geq 0$ and $t_{k}^{\mathrm{ul}} \geq 0$. Therefore, both the problems are equivalent. Then, the continuous-binary product, can be handled by the following substitution
\begin{align}
    z^{\mathrm{dl}}_{kj} = d_{kj} t_k^{\mathrm{dl}}, z^{\mathrm{ul}}_{kj} = u_{kj} t_k^{\mathrm{ul}}, \quad \text{for all $k$ and $j$} \label{eq:18}
\end{align}
where the auxiliary variables are $\bm z^{\mathrm{dl}} = [z^{\mathrm{dl}}_{kj}]_{k=1,j=1}^{k=K, j=J}$ and $\bm z^{\mathrm{ul}} = [z^{\mathrm{ul}}_{kj}]_{k=1,j=1}^{k=K, j=J}$. However, the above substitution is difficult to implement; thus, we use standard big-$M$~\cite{Dejans2025} to enforce the above equality. The constraints are as follows:
\begin{subequations}\label{eq:19}
    \begin{align}
         &0 \leq z_{kj}^{\text{dl}} \leq  t^{\text{dl}}_{k}, 0 \leq z_{kj}^{\text{ul}} \leq  t^{\text{ul}}_{k}, \label{eq:19a}\\
        &t_k^{\text{dl}} - M(1-d_{kj}) \leq z_{kj}^{\text{dl}} \leq Md_{kj}, \label{eq:19b}\\
        &t_k^{\text{ul}} - M(1-u_{kj}) \leq z_{kj}^{\text{ul}} \leq Mu_{kj}.\label{eq:19c}
    \end{align}
\end{subequations}
Note, applying big-$M$, we don't lose any optimality as the above constraints, enforces the equality in \eqref{eq:18} exactly if the value of $M$ chosen properly~\cite{Dejans2025}.
Finally, the problem in $(\mathrm{P3})$ is written with big-$M$ constraints as
\begin{figure*}[b]
        \hrulefill 
		\vspace{0.1cm}
    \begin{minipage}{\textwidth}
        \begin{align}
            f_3(\bm d, \bm r, \bm u, &\bm t^{\text{dl}}, \bm t^{\text{ul}}, \bm z^{\mathrm{dl}}, \bm z^{\mathrm{ul}},\bm \chi^{\text{dl}} , \bm \chi^{\text{ul}}, \bm \xi^{\text{dl}}, \bm \xi^{\text{ul}} ) =   \sum_{k \in \mathcal{K}} \sum_{j\in \mathcal{J}}  \Big[\log_2 (1+  	\chi^{\text{dl}}_{kj} )	+ \log_2 ( 1 + \chi^{\text{ul}}_{kj} ) - \chi^{\text{dl}}_{kj}  - \chi^{\text{ul}}_{kj} - (\xi_{kj}^{\text{dl}} \sigma)^2 - (\xi_{kj}^{\text{ul}} \sigma)^2
             \nonumber \\
            &+  2 \xi_{kj}^{\text{dl}} \, z_{kj}^{\mathrm{dl}}\sqrt{(1+  	\chi^{\text{dl}}_{kj} )} \, \lVert\bm \gamma^{\text{dl}}_{kjj}\rVert + 2 \xi_{kj}^{\text{ul}} \, z_{kj}^{\mathrm{ul}}\sqrt{(1+  	\chi^{\text{ul}}_{kj} )} \, \lVert\bm \gamma^{\text{ul}}_{kjj}  \rVert - (z_{kj}^{\mathrm{dl}})^2 \sum_{j' \in \mathcal{J}} \sum_{k' \in \mathcal{K}} (\xi_{k'j'}^{\text{dl}})^2  \, \lvert (\bm \gamma^{\text{dl}}_{k'j'j})^T \bm w_{kj} \rvert^2    \nonumber \\ 
            & - (z_{kj}^{\mathrm{ul}})^2 \sum_{j' \in \mathcal{J}_{-j}} \sum_{k' \in \mathcal{K}_{-k}} (\xi_{k'j'}^{\text{dl}})^2  \,\lvert \nu_{k'kj'j} \rvert^2- (z_{kj}^{\mathrm{ul}})^2  \sum_{j' \in \mathcal{J}} \sum_{k' \in \mathcal{K}} (\xi_{k'j'}^{\text{ul}})^2   \lvert \bm v^T_{k'j'} \bm \gamma^{\text{ul}}_{kj'j} \rvert^2   \Big]. \label{eq:20}
        \end{align}
    \end{minipage}
\end{figure*}
\begin{subequations}\label{eq:20}
		\begin{align}
		(\mathrm{P4): } &\underset{ \substack{\bm d,  \bm u, \bm t^{\text{dl}}, \bm t^{\text{ul}}, \bm z^{\mathrm{dl}}, \bm z^{\mathrm{ul}}}}{\text{max}}  f_3( \bm d, \bm r, \bm u, \bm t^{\text{dl}}, \bm t^{\text{ul}}, \bm z^{\mathrm{dl}},  \bm z^{\mathrm{ul}}, \bm \chi^{\text{dl}\star} , \nonumber \\ 
        &\hspace{135pt} \bm \chi^{\text{ul}\star}, \bm \xi^{\text{dl}\star},   \bm \xi^{\text{ul}\star}) \nonumber\\
		& \text{s.t. } \Big\lvert \sum_{j \in \mathcal{J}} d_{kj} r_{j} - \sum_{j \in \mathcal{J}}  u_{kj} r_{j} \Big \rvert  \nonumber \\
		&\quad \qquad \qquad  \leq M( 2- \sum_{j \in \mathcal{J}} d_{kj} - \sum_{j \in \mathcal{J}} u_{kj} ), \forall k \label{eq:20a} \\
        & \qquad \sum_{k \in \mathcal{K}} (z_{kj}^{\text{dl}})^2\leq p_{j}^{\text{max}}, \text{ for all } j\\
		& \qquad \eqref{eq:8c}, \eqref{eq:8e}-\eqref{eq:8g}, \eqref{eq:19a}-\eqref{eq:19c}.
	\end{align}
\end{subequations}
where in \eqref{eq:20a}, we enforce the same logic in \eqref{eq:8b} by a large number $M$ for ease of implementation in  solvers. The objective $f_3$ is shown at the bottom in \eqref{eq:20}, which is concave in $z_{kj}^{\mathrm{dl}}$ and $z_{kj}^{\mathrm{ul}}$. Further, continuous relaxation of feasible set is convex. Therefore, the program can be solved to optimality with standard techniques, e.g., branch \& bound, and convex optimization, implemented in industry grade solvers, e.g., MOSEK, which is used here.

As mentioned before, the value of $M$ must be chosen properly. First, in constraint \eqref{eq:19}, a suitable value is $M \geq \max_{k,j} \{\sqrt{p_j^{\text{max}}}, \sqrt{p_k^{\text{max}}} \}$~\cite{Dejans2025}. Additionally, we also used $M$ for conditional switch in \eqref{eq:20a}, for this $M$ must be $\geq \max_k \, \lvert \sum_{j \in \mathcal{J}} d_{kj} r_{j} - \sum_{j \in \mathcal{J}}  u_{kj} r_{j} \rvert $; whose maximum value is $1$, due to constraint~\eqref{eq:8c}. Combining two conditions, we choose, $M = \lceil \max \{ \max_{k,j} \{\sqrt{p_j^{\text{max}}}, \sqrt{p_k^{\text{max}}} \}, 1\} \rceil$, where $\lceil \cdot \rceil$ denotes the smallest integer function.

Finally, the complete solution procedure is outlined in Algorithm~\ref{algo:1}. The algorithm iterates over all binary combinations of $\bm r$. For each fixed $\bm r$, the variables $\bm d, \bm u, \bm p^{\text{dl}}, \bm p^{\text{ul}}$ are initialized to satisfy \eqref{eq:8b}–\eqref{eq:8g}. Then, the variables are updated block-wise: step-4: update $\bm \chi^{\text{dl}}, \bm \chi^{\text{ul}}$ by \eqref{eq:14}; step-5: update $\bm \xi^{\text{dl}}, \bm \xi^{\text{ul}}$ by \eqref{eq:13}; and step-6: solve $(\mathrm{P4})$ to obtain $\bm d, \bm u, \bm p^{\text{dl}}, \bm p^{\text{ul}}$. These updates repeat until convergence for each $\bm r$, after convergence, select
$\big(\bm d^\dagger, \bm r^\dagger, \bm u^\dagger, \bm p^{\mathrm{dl}\dagger}, \bm p^{\mathrm{ul}\dagger}, \bm \chi^{\mathrm{dl}\dagger}, \bm \chi^{\mathrm{ul}\dagger}, \bm \xi^{\mathrm{dl}\dagger}, \bm \xi^{\mathrm{ul}\dagger}\big)$ 
corresponding to the highest achieved objective value $f_2$. Then, the optimal $f_0$ is obtained from~\eqref{eq:7}.

\begin{proposition}
    Algorithm~\ref{algo:1} converges in a finite number of iterations. Further, it has at least one limit point, and every limit point is a stationary point for $\bm \chi^{\mathrm{dl}}$, $\bm \chi^{\mathrm{ul}}$, $\bm \xi^{\mathrm{dl}}$, and $\bm \xi^{\mathrm{ul}}$.
\end{proposition}
\begin{proof}
    See Appendix~\ref{appendix:2}.
\end{proof}

\begin{algorithm}[t]
	\caption{Solution for optimization $(\mathrm{P0})$}
	\begin{algorithmic}[1]
         \ForAll{$\bm r = [r_1, \cdots, r_J] \in \{0,1\}^J$}
            \State \textbf{Initialize:} feasible values of $\bm d, \bm u, \bm p^{\text{dl}}, \bm p^{\text{ul}}$
            \Repeat
                \State  update $\bm \chi^{\text{dl}}, \bm \chi^{\text{ul}}$ by \eqref{eq:14}
                \State update $\bm \xi^{\text{dl}}, \bm \xi^{\text{ul}}$ by \eqref{eq:13}
                \State solve $(\mathrm{P4})$ to get $\bm d, \bm u, \bm p^{\text{dl}}, \bm p^{\text{ul}}$
            \Until{convergence}
         \EndFor
         \State select 
$\big(\bm d^\dagger, \bm r^\dagger, \bm u^\dagger, \bm p^{\mathrm{dl}\dagger}, \bm p^{\mathrm{ul}\dagger}, \bm \chi^{\mathrm{dl}\dagger}, \bm \chi^{\mathrm{ul}\dagger}, \bm \xi^{\mathrm{dl}\dagger}, \bm \xi^{\mathrm{ul}\dagger}\big)$ 
corresponding to the highest achieved objective value $f_2$ over all feasible $\bm r$. Then, the optimal objective $f_0$ is obtained from~\eqref{eq:7} using $(\bm d^\dagger, \bm r^\dagger, \bm u^\dagger, \bm p^{\mathrm{dl}\dagger}, \bm p^{\mathrm{ul}\dagger})$. 
	\end{algorithmic}
    \label{algo:1}
\end{algorithm}

\section{Numerical Results}
\label{sec:4}


In this section, we present the numerical results corresponding to the above formulations. We consider a scenario where $K=10$ randomly distributed \cglspl{ue} are placed within a circular region of radius 100~\si{m} centered around location $(53.0793^\circ\mathrm{N},\, 8.8017^\circ\mathrm{E})$. These \cglspl{ue} are served by $J$ satellites. The frequency bands are $f_1 = 2.4$~\si{GHz} and $f_2 = 1.9$~\si{GHz}, with $B_1 = B_2 = 10$~\si{MHz}. Each satellite is equipped with $N = 16 \times 16$ antennas arranged in a \gls{upa} configuration. The inter-\cgls{ue} channels are modeled as purely \cgls{los}. The maximum transmit power at each satellite is 20~\si{W}, while the \cglspl{ue} transmit with a maximum power of 2~\si{W} in the \cgls{ul} direction. The satellites are positioned at an altitude of 500~\si{km}, each deployed at different elevation and inclination angles.

Figure~\ref{fig:4} illustrates the evolution of the objective  $f_2$ with iterations for $J = [2, 3, 4]$. It can be observed that the proposed algorithm converges within a finite number of steps, consistent with the convergence result established in Proposition~2. 

Figure~\ref{fig:5} shows the \gls{cdf} of the sum-rate for $J=[2,3,4]$. In the proposed setup, the ``with spin'' configuration implies that each satellite $j$ adaptively selects its $r_j$ through the optimization. Conversely, in the ``without spin'' case, we fix $r_j = 0$ for all satellites $j$, which corresponds to a conventional non-adaptive configuration. It can be observed that for $J=2$, both configurations exhibit nearly identical performance. This behavior arises from the strong inter-\cgls{ue} interference that occurs when the \cglspl{ue} operate in opposite spin states. Since the inter-\cgls{ue} channels are modeled as purely \cgls{los}, the benefits of frequency spin are effectively neutralized in this case. However, as the number of satellites increases, the interference level also grows, and the proposed ``with spin'' strategy demonstrates significant gains, approximately 78~\% improvement for $J=3$ and nearly 94~\% improvement for $J=4$ in the tail of the distribution, as illustrated in the figure. 

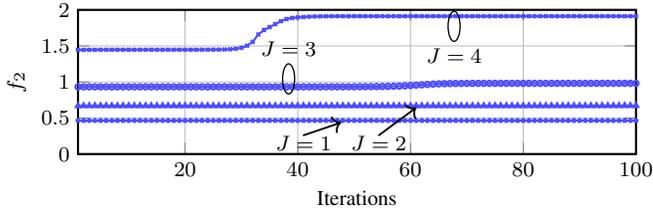
\begin{figure}
    \centering
    \footnotesize{
\begin{tikzpicture}
\begin{axis}[
    width=9cm,
    height=3.5cm,
    xlabel={Iterations},
    ylabel={$f_2$},
    xmin=1, xmax=100,
    ymin = 0, ymax = 2,
    grid=major,
    legend style={at={(0.5,-0.15)}, anchor=north, legend columns=-1},
    thick,
]

\addplot[
    blue!70,
    mark=star,
    mark size=1pt,
    mark options={solid},
]
table [x expr=\coordindex+1, y index=0] {Data/f2_J_1_K_10.txt};

\addplot[
    blue!70,
    mark=triangle,
    mark size=1pt,
    mark options={solid},
]
table [x expr=\coordindex+1, y index=0] {Data/f2_J_2_K_10.txt};

\addplot[
    blue!70,
    mark=o,
    mark size=1pt,
    mark options={solid},
]
table [x expr=\coordindex+1, y index=0] {Data/f2_J_3_K_10.txt};

\addplot[
    blue!70,
    mark=x,
    mark size=1pt,
    mark options={solid},
]
table [x expr=\coordindex+1, y index=0] {Data/f2_J_4_K_10.txt};

\end{axis}

\draw[line width=0.5pt, line cap=round] (5,1.7) ellipse (.08cm and 0.2cm);
\node[below=0.2cm] at (5,1.7) {\shortstack{$J = 4$}};

\draw[line width=0.5pt, line cap=round] (2.8,1) ellipse (.08cm and 0.2cm);
\node[above=0.2cm] at (2.8,1) {\shortstack{$J = 3$}};

\draw[->, line width=0.7pt, line cap=round] (4,0.25) -- (4.5,0.62)
    node[pos=0, below=-2.9pt] {$J = 2$};

\draw[->, line width=0.7pt, line cap=round] (3,0.25) -- (3.5,0.42)
    node[pos=0, below=-2.9pt] {$J = 1$};    
\end{tikzpicture}}
\caption{Monotonic increase of $f_2$ with iterations for the proposed Algorithm~\ref{algo:1}.}

    \label{fig:4}
\end{figure}

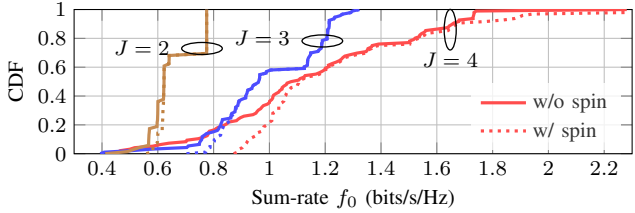
\begin{figure}
    \centering
    \begin{tikzpicture}
    \footnotesize{
\begin{axis}[
    width=9cm,
    height=3.5cm,
    xlabel={Sum-rate $f_0$ (bits/s/Hz)},
    ylabel={CDF},
    xmin =0.3, xmax = 2.3,
    ymin =0, ymax = 1,
    grid=major,
    legend style={at={(0.98,0.02)}, anchor=south east, draw=none, fill=white, fill opacity=0.8},
    legend cell align={left}
]

\addplot+[
      mark=none,
      very thick,
      red!70,
    ] 
    table [
      x expr=\thisrowno{0},
      y expr=\coordindex / (96),
      col sep=space,
    ] {Data/J_4_K_10_n.txt};

\addplot+[
      mark=none,
      very thick,
      dotted,
      red!70,
    ] 
    table [
      x expr=\thisrowno{1},
      y expr=\coordindex / (96),
      col sep=space,
    ] {Data/J_4_K_10_n.txt};

\addplot+[
      mark=none,
      very thick,
      blue!70,
    ] 
    table [
      x expr=\thisrowno{0},
      y expr=\coordindex / (93),
      col sep=space,
    ] {Data/J_3_K_10_n.txt};

\addplot+[
      mark=none,
      very thick,
      dotted,
      blue!70,
    ] 
    table [
      x expr=\thisrowno{1},
      y expr=\coordindex / (93),
      col sep=space,
    ] {Data/J_3_K_10_n.txt};

\addplot+[
      mark=none,
      very thick,
      brown!90,
    ] 
    table [
      x expr=\thisrowno{0},
      y expr=\coordindex / (69),
      col sep=space,
    ] {Data/J_2_K_10_n.txt};

\addplot+[
      mark=none,
      very thick,
      dotted,
      brown!100,
    ] 
    table [
      x expr=\thisrowno{1},
      y expr=\coordindex / (69),
      col sep=space,
    ] {Data/J_2_K_10_n.txt};

\addlegendimage{area legend, draw=none} 
\addlegendimage{line legend, thick, black, solid}
\addlegendentry{w/o spin}
\addlegendimage{line legend, thick, black, dashed}
\addlegendentry{w/ spin}

\end{axis}

\draw[line width=0.5pt, line cap=round] (5,1.65) ellipse (.08cm and 0.27cm);
\node[below=0.2cm] at (5,1.65) {\shortstack{$J = 4$}};

\draw[line width=0.5pt, line cap=round] (1.7,1.4) ellipse (.27cm and 0.08cm);
\node[below=0.3cm] at (.9,1.9) {\shortstack{$J = 2$}};

\draw[line width=0.5pt, line cap=round] (3.3,1.5) ellipse (.27cm and 0.08cm);
\node[above=0.35cm] at (2.5,1) {\shortstack{$J = 3$}};
}
\end{tikzpicture}
   \caption{CDF of the sum-rate $f_0$ with and without spin. }

    \label{fig:5}
\end{figure}

\section{Conclusions}  
\label{sec:5}
This paper presented a dynamic band allocation framework for LEO satellite systems in 6G networks. By enabling satellites and users to jointly select \cgls{dl} and \cgls{ul} frequency bands, the approach adds flexibility to conventional FDD systems and enhances spectrum utilization. A joint optimization problem was formulated for dynamic band selection, user scheduling, and power allocation, and solved through efficient transformation-based methods. Simulation results showed that the proposed strategy effectively mitigates interference in dense constellations, achieving up to 94\% performance improvement in high-interference scenarios.

\appendix
\subsection{Proof of \textup{\textbf{Proposition 1}}}
\label{appendix:1}
Observe that $f_2$ is strictly concave in  $\bm \xi^{\mathrm{dl}}$ and in $\bm \xi^{\mathrm{ul}}$. Therefore,  keeping other variables fixed, unique optimal point $\xi^{\mathrm{dl}\star}_{kj}$ and $\xi^{\mathrm{ul}\star}_{kj}$ is obtained from $\partial f_2 / \partial \xi^{\text{dl}}_{kj} = 0$ and $\partial f_2 / \partial \xi^{\text{ul}}_{kj} = 0$. Substituting this in $f_2$ recovers $f_1$ exactly. Similarly, $f_1$ is strictly concave in $\bm \chi^{\mathrm{dl}}$ and  $\bm \chi^{\mathrm{ul}}$. Therefore, unique optimal 
$\chi^{\mathrm{dl}\star}_{kj}$ and $\chi^{\mathrm{ul}\star}_{kj}$ can be obtained from $\partial f_1 / \partial \chi^{\text{dl}}_{kj} = 0$ and $\partial f_1 / \partial \chi^{\text{ul}}_{kj} = 0$. Substituting this in $f_1$ gives $f_0$ exactly.
\qed

\subsection{Proof of \textup{\textbf{Proposition 2}}}

\label{appendix:2}
We introduce a subscript $t$ for the iteration index, and we ignore the variable $\bm r$, as we search over all possible values. Therefore, for the objective $f_2$ the following holds:
\begin{align*}
 &f_2\Big(\bm d_{(t)},  \bm u_{(t)}, \bm p^{\mathrm{dl}}_{(t)}, \bm p^{\mathrm{ul}}_{(t)}, \,\bm \chi^{\mathrm{dl}}_{(t)}, \bm \chi^{\mathrm{ul}}_{(t)}, {\bm \xi}^{\mathrm{dl}}_{(t)}, {\bm \xi}^{\mathrm{ul}}_{(t)}\Big) \\
 &\overset{(a)}{=} f_1\Big(\bm d_{(t)},  \bm u_{(t)}, \bm p^{\mathrm{dl}}_{(t)}, \bm p^{\mathrm{ul}}_{(t)}, \,\bm \chi^{\mathrm{dl}}_{(t)}, \bm \chi^{\mathrm{ul}}_{(t)}\Big) \\
 &\overset{(b)}{<} f_1\Big(\bm d_{(t)},  \bm u_{(t)}, \bm p^{\mathrm{dl}}_{(t)}, \bm p^{\mathrm{ul}}_{(t)}, \,\bm \chi^{\mathrm{dl}}_{(t+1)}, \bm \chi^{\mathrm{ul}}_{(t+1)}\Big) \\
 &  \overset{(c)}{=}f_2\Big(\bm d_{(t)},  \bm u_{(t)}, \bm p^{\mathrm{dl}}_{(t)}, \bm p^{\mathrm{ul}}_{(t)}, \,\bm \chi^{\mathrm{dl}}_{(t+1)}, \bm \chi^{\mathrm{ul}}_{(t+1)}, \bm \xi^{\mathrm{dl}}_{(t+1)}, \bm \xi^{\mathrm{ul}}_{(t+1)}\Big) \\
 &\overset{(e)}{\leq}f_2\Big(\bm d_{(t+1)},  \bm u_{(t+1)}, \bm p^{\mathrm{dl}}_{(t+1)}, \bm p^{\mathrm{ul}}_{(t+1)}, \,\bm \chi^{\mathrm{dl}}_{(t+1)}, \bm \chi^{\mathrm{ul}}_{(t+1)}, \bm \xi^{\mathrm{dl}}_{(t+1)}, \\
 & \hspace{217pt} \bm \xi^{\mathrm{ul}}_{(t+1)}\Big)
\end{align*}
where $(a)$ follows from proof of Proposition~1, as keeping other variables fixed, $f_2$ is strictly concave in ${\bm \xi}^{\mathrm{dl}}$ and ${\bm \xi}^{\mathrm{ul}}$. Therefore, we can obtain ${\bm \xi}^{\mathrm{dl}}$ and ${\bm \xi}^{\mathrm{ul}}$ uniquely as in \eqref{eq:13} and substituting these in $f_2$ we get $f_1$ exactly. Further, $(b)$ follows as $f_1$ is strictly concave in ${\bm \chi}^{\mathrm{dl}}$ and ${\bm \chi}^{\mathrm{ul}}$; therefore, until a stationary point is obtained, the objective strictly increases with each update. Then, $(c)$ follows similar logic as in $(a)$. Finally, $(e)$ follows, as keeping other variables fixed, joint update of $(\bm d,  \bm u, \bm p^{\mathrm{dl}}, \bm p^{\mathrm{ul}})$ maximizes the objective.  Also, as the transmit powers are fixed, and the each update strictly increases the objective, the algorithm converges in finite steps.
\qed

\bibliographystyle{ieeetr}
\bibliography{ref}

\end{document}

%% file: acronyms.tex
\makeglossaries

\newacronym{3gpp}{3GPP}{$\text{3}^{\text{rd}}$ generation partnership project}
\newacronym{5g}{5G}{fifth-generation}
\newacronym{6g}{6G}{sixth-generation}

\newacronym{bs}{BS}{base-station}

\newacronym{cdf}{CDF}{cumulative distribution function}
\newacronym{csi}{CSI}{channel state information}

\newacronym{drl}{DRL}{deep-reinforcement learning}
\newacronym{dl}{DL}{downlink}

\newacronym{ecef}{ECEF}{earth-centered, earth-fixed}
\newacronym{efc}{EFC}{Earth-fixed cell}
\newacronym{emc}{EMC}{Earth-moving cell}

\newacronym{fdd}{FDD}{frequency-division duplex}
\newacronym{fspl}{FSPL}{free-space path loss}

\newacronym{gnb}{gNB}{next-generation node-B}
\newacronym{gnss}{GNSS}{global navigation satellite system}
\newacronym{geo}{GEO}{geostationary Earth orbit}
\newacronym{iot}{IoT}{internet-of-things}

\newacronym{leo}{LEO}{low Earth orbit}
\newacronym{los}{LOS}{line-of-sight}

\newacronym{mab}{MAB}{multi-arm bandit}
\newacronym{meo}{MEO}{medium Earth orbit }
\newacronym{mimo}{MIMO}{multiple-input multiple-output}
\newacronym{mrc}{MRC}{maximal-ratio combiner}
\newacronym{mrt}{MRT}{maximal-ratio transmission}

\newacronym{ned}{NED}{North-East-Downward}
\newacronym{neu}{NEU}{North-East-Upward}
\newacronym{ntn}{NTN}{non-terrestrial network}
\newacronym{nlos}{NLOS}{non-line-of-sight}
\newacronym{ngeo}{NGEO}{non-geostationary Earth orbit}

\newacronym{ofdm}{OFDM}{orthogonal-frequency division multiplexing}

\newacronym{ran}{RAN}{radio-access network}

\newacronym{sa}{SA}{simulated annealing}
\newacronym{sinr}{SINR}{signal-to-noise plus interference ratio}

\newacronym{tdd}{TDD}{time-division duplex}

\newacronym{ue}{UE}{user-equipment}
\newacronym{ul}{UL}{uplink}
\newacronym{upa}{UPA}{uniform-planar array}